\documentclass[letterpaper, 10 pt, conference]{ieeeconf}  % Comment this line out if you need a4paper

\IEEEoverridecommandlockouts                              % This command is only needed if 
\usepackage{graphics} % for pdf, bitmapped graphics files
\usepackage{epsfig} % for postscript graphics files
\usepackage{times} % assumes new font selection scheme installed

\usepackage{color}
\usepackage{subfig}
\usepackage[hidelinks]{hyperref}

\usepackage{amsmath} 
\usepackage{amssymb}
\usepackage{amsthm}

\usepackage{mathtools}
\mathtoolsset{showonlyrefs=true}

\usepackage{empheq}

\newcommand{\RealSet}{\mathbb{R}}
\newcommand{\diag}[1]{\operatorname{diag}\left\{#1\right\}}		% Diagonal concatenation
\newcommand{\ver}[1]{\operatorname{ver}\left\{#1\right\}}		% Vertical concatenation
\newcommand{\supp}[1]{\operatorname{supp}\left\{#1\right\}}		% Support set
\newcommand{\suppeps}[2]{\operatorname{supp}_{#1}\left\{#2\right\}}		% Support set

\DeclareMathOperator*{\argmin}{arg\,min}						% argmin

\newtheorem{proposition}{Proposition}
\newtheorem{theorem}{Theorem}
\newtheorem{assumption}{Assumption}
\newtheorem{remark}{Remark}

\newcommand{\q}{\mathbf{q}}
\newcommand{\p}{\mathbf{p}}

\newcommand{\h}{\mathbf{h}}
\newcommand{\x}{\mathbf{x}}

\newcommand{\de}{\text{d}}

\usepackage{acronym}
\acrodef{ADF}{\emph{agent descriptor function}}
\acrodef{CTDF}{\emph{current task descriptor function}}
\acrodef{TDF}{\emph{desired task descriptor function}}
\acrodef{TEF}{\emph{task error function}}

\usepackage{setspace}

\title{\LARGE \bf
	Persistent Coverage Control for Teams of Heterogeneous Agents*\\ -- Extended Version --
}

\author{
	Alberto Mellone$^{1,2}$, Giovanni Franzini$^{1,3}$, Lorenzo Pollini$^{1}$ and Mario Innocenti$^{1}$%
	\thanks{
		*This paper is an extended version of: ``A. Mellone, G. Franzini, L. Pollini and M. Innocenti, Persistent coverage control for teams of heterogeneous agents, to appear in \textit{Proc. 57th IEEE Conference on Decision and Control}, Miami Beach, FL, USA, Dec. 2018''.
	}%
	\thanks{
		*This work was partially supported by the University of Pisa PRA-2017-15 and PRA-2017-41 research grants.
	}%
	\thanks{$^{1}$University of Pisa, Department of Information Engineering, Largo L. Lazzarino~1, 56122 Pisa, Italy.
	}%
	\thanks{$^{2}$A. Mellone is currently with Imperial College London, Department of Electrical and Electronic Engineering, SW7 2AZ, London, UK. The work presented herein was carried out when he was Graduate Student at University of Pisa.
	}%
	\thanks{$^{3}$G. Franzini is currently with United Technologies Research Centre Ireland, 4th Floor, Penrose Business Center, Penrose Wharf, T23 XN53, Cork, Republic of Ireland. The work presented herein was carried out when he was Ph.D. Candidate at University of Pisa.
	}%
	\thanks{E-mail addresses: {\itshape\small a.mellone18@imperial.ac.uk} (A. Mellone), {\itshape\small franzigi@utrc.utc.com} (G. Franzini), {\itshape\small lorenzo.pollini@unipi.it} (L. Pollini), {\itshape\small mario.innocenti@unipi.it} (M. Innocenti).}%
}

\begin{document}

\maketitle
\thispagestyle{empty}
\pagestyle{empty}

%%%%%%%%%%%%%%%%%%%%%%%%%%%%%%%%%%%%%%%%%%%%%%%%%%%%%%%%%%%%%%%%%%%%%%%%%%%%%%%%
\begin{abstract}

A distributed cooperative control law for persistent coverage tasks is proposed, capable of coordinating a team of heterogeneous agents in a structured environment. Team heterogeneity is considered both at vehicles' dynamics and at coverage capabilities levels. More specifically, the general dynamics of nonholonomic vehicles are considered. Agent heterogeneous sensing capabilities are addressed by means of the descriptor function framework, a set of analytical tools for controlling agents involved in generic coverage tasks. By means of formal arguments, we prove that the team performs the task and no collision occurs between agents nor with obstacles. A numerical simulation validates the proposed strategy.

\end{abstract}

%%%%%%%%%%%%%%%%%%%%%%%%%%%%%%%%%%%%%%%%%%%%%%%%%%%%%%%%%%%%%%%%%%%%%%%%%%%%%%%%
\section{Introduction}

Distributed robotic systems are currently one of the most extensively studied topics worldwide. Besides the theoretical implications, the practical advantages they yield make them particularly attractive. Indeed, a team of autonomous agents, each one characterized by a reduced computational and implementation complexity, can enforce local, basic actions to pursue a global desired behavior in a cost-efficient way. Coverage problems are one of the branches of multi-agent systems study, featuring autonomous vehicles exploring the environment to gather data through on-board sensors or to provide resources to different regions. Applications are extremely differentiated, from surveillance to geographical surveys, lawn mowing, floor cleaning, and environmental monitoring~\cite{Sahin2007,Dunbabin2012,Avellar2015}.

The coverage type considered in this paper is known as \emph{persistent coverage} and consists in agents recursively exploring the environment, thus modeling scenarios in which information that has been collected in the past should be updated over time. Alternately, resources previously delivered are supposed to undergo a degradation process and, in turn, require the team to sweep again already visited spots.

In the literature, different types of techniques have been used to address this problem. Solutions based on waypoint guidance have been proposed for agents with single-integrator~\cite{Hokayem2007,Franco2013} and unicycle kinematics~\cite{Franco2015}. Path planning techniques are proposed in~\cite{Song2013,PalaciosGasos2017}. In the former reference, the authors discuss the geometric design of paths aimed at performing the task, that the agents must follow, whereas in~\cite{PalaciosGasos2017} the single-integrator kinematics of the agents is explicitly taken into account in the planning. A gradient based law is instead proposed in~\cite{Wang2010} for a team of single-integrator agents. Event-driven strategies for single-integrator agents are proposed in~\cite{Zhou2016}. The majority of the proposed solutions consider teams composed by agents with the same kinematics, and with sensing capabilities characterized by the same mathematical model, although the parameters may vary from agent to agent. The control of a team of heterogeneous agents seems to be still un-addressed in the literature, at the best of the authors' knowledge.

The main contribution of the paper is a distributed cooperative control law for persistent coverage tasks, designed for coordinating
the motion of a team composed by agents with general sensing capabilities and different dynamics, extending the results presented in~\cite{Franzini2017}. 

More specifically, we consider agents with general nonholonomic dynamics, thus providing a control law that better complies with a rather wide class of vehicles. Sensing heterogeneity is addressed by means of the analytical tools offered by the \emph{descriptor function framework}. Introduced in~\cite{Niccolini2010}, the framework addresses the need for a unified method to approach different kinds of coverage tasks. It provides a mathematical abstraction, the \emph{descriptor function}, which models both the task requirements and the agents' contributions. This framework allows for an abstraction from the actual coverage capabilities of each agent, as well as from the actual scenario. A control law is then designed to make the agents cooperatively carry out the assigned task by solely resorting to information retrieved through local sensors or inter-agent communication. In addition, a unified approach to avoid collisions between the agents and with the obstacles is provided, ensuring the safe execution of the task. The validity and efficacy of the proposed technique is assessed by means of formal arguments and is shown through a simulation.

\paragraph*{Notation}

The field of reals is denoted with $\RealSet$. The following subsets of $\RealSet$ are introduced: $\RealSet_0^+ = \left\{x\in \RealSet \colon x\ge 0 \right\}$, and $\RealSet^+ = \RealSet_0^+ \backslash \{0\}$. The $n \times n$ identity matrix is denoted with $\mathbf{I}_n$. Given a positive semi-definite matrix $\mathbf{A} \in \RealSet^{n \times n}$, the weighted Euclidean norm of $\mathbf{v} \in \RealSet^n$ is denoted with $\Vert \mathbf{v} \Vert_{\mathbf{A}} = \sqrt{\mathbf{v}^T\mathbf{A}\mathbf{v}}$. The support operator $\supp{\cdot }$ of a real- or vector-valued function $\mathbf{f}(\x)$ is defined as
$\supp{\mathbf{f}} = \left\{ \x \colon \mathbf{f}(\x) \ne \mathbf{0} \right\}$, i.e. it is the set of points of its domain where $\mathbf{f}$ is not identically zero. The operators $\ver{\cdot}$ and $\diag{\cdot}$ represent the vertical and block diagonal concatenations.

\section{Team Modeling}

\subsection{Preliminaries}

Consider a \emph{team} $\mathcal{T}$ composed by $N$ agents, allowed to operate in a closed and bounded topological set $Q \subset \RealSet^n$, with $n \in \{2,3\}$. We denote with $\p_i \in \mathcal{C}(Q)$ the $i$-th agent \emph{pose}, i.e. its position and orientation in $Q$, where $\mathcal{C}(Q)$ represents the agent configuration space on $Q$. 

The operational area may be populated by \emph{obstacles}. Without loss of generality, obstacles will be considered convex polygons ($n = 2$) or polyhedra ($n = 3$). The $k$-th obstacle occupies the region $Q_{\text{obs}}^k \subseteq Q$. Given the obstacle set $\mathcal{O}$, we define $Q_{\text{obs}}=\bigcup_{k \in \mathcal{O}} Q_{\text{obs}}^k$.

\subsection{Descriptor Function Framework Overview}

A brief overview of the descriptor function framework is now provided, reviewing the elements essential for the paper. The interested reader shall refer to~\cite{Niccolini2011} for further details. The main idea behind the framework is the use of a common abstraction, the \emph{descriptor function}, for modeling both the agent capabilities and the deployment requested to the team in order to accomplish the task.

The \ac{ADF} $d_i \colon \mathcal{C}(Q)\times Q \to \RealSet_0^+$ represents the agent $i\in\mathcal{T}$ capability of performing the assigned task or, equivalently for a coverage task, the amount of sensing that it instantly provides. The \ac{ADF} is assumed to be continuous and differentiable over $Q$, its support being a connected set. For \ac{ADF}s with unbounded support it is reasonable to define:
\begin{equation}
	\suppeps{\epsilon}{ d_i(\p_i, \cdot) } 
	= 
	\left\{
		\q\in Q \colon d_i(\p_i, \q) > \epsilon ,\, \epsilon \in \RealSet_0^+
	\right\}
\end{equation}     
In the remainder, the distinction between spatially bounded and unbounded \ac{ADF}s will not be explicitly addressed, and, with a slight abuse of notation, $\supp{ \cdot}$ will also denote $\suppeps{\epsilon}{\cdot}$ in the latter case.

\begin{remark}
	\label{remark:support_inclusion}
	Given the positive definiteness, continuity and differentiability of the \ac{ADF}s, the following result holds:
	\begin{equation}
		\supp{\frac{\partial d_i}{\partial \p_i} } \subseteq \supp{d_i}
	\end{equation}
\end{remark}

The sum of all the \ac{ADF}s is called the \ac{CTDF} and represents the cumulative amount of sensing that the team is instantly achieving. It is denoted with $d \colon \mathcal{C}(Q)^N\times Q \to \RealSet_0^+$, and is defined as:
\begin{equation}
\label{eq:cTDF}
	d(\p, \q) = \sum_{i \in \mathcal{T}} d_i(\p_i, \q),
	\quad
	\p = \ver{\p_i} \in \mathcal{C}(\RealSet^n)^N
\end{equation} 

The \ac{TDF} $d_* \colon \RealSet_0^+ \times Q \to \RealSet_0^+$ describes how the agents should be distributed in the operational area in order to maximize the goal achievement. Therefore, $d_*(t,\q)$ defines how much of the available sensing capability is needed at time $t$ at point $\q$. With the definition of the \ac{CTDF} and of the \ac{TDF}, it is natural to introduce an error between the amount of sensing that the task requires and that is actually provided by the team at each $\q \in Q$. This error is quantified by the \ac{TEF}, defined as:
\begin{equation}
\label{eq:TEF}
	e(t, \p, \q) = d_*(t, \q) - d(\p,\q)
\end{equation}
The \ac{TEF} models the excess or the lack of sensing over time at each point of the environment.

\subsection{Agents Dynamics}
\label{sec:agents_motion}

We consider the following dynamics for the $i$-th agent:
\begin{subequations}
\label{eq:general_dynamics}
\noeqref{eq:dynamic_constrained,eq:kinematic_constraints}
	\begin{empheq}[left=\empheqlbrace\,]{align}
	\label{eq:dynamic_constrained}
	&\mathbf{B}_i(\x_i)\ddot{\x}_i 
	+ 
	\mathbf{c}_i(\x_i, \dot{\x}_i) 
	= 
	\mathbf{D}_i(\x_i)\boldsymbol{\tau}_i 
	+ 
	\mathbf{A}_i(\x_i)\boldsymbol{\lambda}_i
	\\
	\label{eq:kinematic_constraints}
	&\mathbf{A}_i(\x_i)^T\dot{\x}_i 
	= 
	\mathbf{0} 
	\end{empheq}
\end{subequations}
where $\x_i \in \RealSet^{q_i}$ is the state vector, $\mathbf{B}_i \colon \RealSet^{q_i} \to \RealSet^{q_i\times q_i}$ is the inertia matrix, $\mathbf{c}_i \colon \RealSet^{q_i}\times\RealSet^{q_i} \to \RealSet^{q_i}$ is the vector gathering the centripetal, Coriolis and potential terms, $\mathbf{D}_i \colon \RealSet^{q_i} \to \RealSet^{q_i\times m_i}$ transforms the inputs $\boldsymbol{\tau}_i\in \RealSet^{m_i}$ to generalized forces, $\mathbf{A}_i \colon \RealSet^{q_i} \to \RealSet^{q_i \times q_i-m_i}$ is the matrix associated to the Pfaffian representation of the $q_i - m_i$ kinematic constraints and $\boldsymbol{\lambda}_i\in \RealSet^{q_i-m_i}$ is the vector of Lagrange multipliers~\cite{Siciliano2009}.

The agent's pose vector is extracted from the state using a continuous and differentiable map $\h_i:\RealSet^{q_i} \rightarrow \mathcal{C}(\RealSet^n)$:
\begin{equation}
\label{eq:agent_pose}
	\p_i = \h_i(\x_i)
\end{equation}
In the rest of the paper we will consider the following equivalent representation of the dynamics in~\eqref{eq:general_dynamics}:
\begin{equation}
\label{eq:agent_dynamics}
	\dot{\mathbf{z}}_i 
	= 
	\begin{bmatrix}
		\dot{\x}_i \\ \dot{\mathbf{v}}_i
	\end{bmatrix} 
	= 
	\begin{bmatrix}
		\mathbf{G}_i(\x_i)\mathbf{v}_i \\ \mathbf{0}
	\end{bmatrix}
	+
	\begin{bmatrix}
		\mathbf{0} \\ \mathbf{I}_m
	\end{bmatrix}
	\mathbf{u}_i
\end{equation}
where $\mathbf{v}_i \in \RealSet^{m_i}$ is known as \emph{pseudo-velocity vector}, $\mathbf{u}_i \in \RealSet^{m_i}$ is the new system input, known as \emph{pseudo-acceleration vector}, and $\mathbf{G}_i: \RealSet^{q_i} \rightarrow \RealSet^{q_i\times m_i}$ has columns that span the null space of $\mathbf{A}_i$, i.e. $\mathbf{A}_i^T\mathbf{G}_i = \mathbf{0}$. 

\begin{proposition}[\cite{Siciliano2009}, Chap. 11.4]
	Assume that $\x_i$ and $\mathbf{v}_i$ are measurable, and that $\mathbf{G}_i(\x_i)^T\mathbf{D}_i(\x_i)$ is invertible for all $\x_i$. Then the dynamics~\eqref{eq:agent_dynamics} is equivalent to~\eqref{eq:general_dynamics} under the nonlinear state-feedback law:
	\begin{equation}
	\label{eq:linearizing_feedback}
		\boldsymbol{\tau}_i 
		=  
		\left( \mathbf{G}_i(\x_i)^T\mathbf{D}_i(\x_i) \right)^{-1} 
		\!
		\left(
			\widetilde{\mathbf{B}}_i(\x_i)\mathbf{u}_i 
			+ 
			\widetilde{\mathbf{c}}_i(\x_i, \mathbf{v}_i)
		\right)
	\end{equation}
	where
	\begin{align}
	\label{eq:c_tilde}
	\!\!\!\!\!
		\widetilde{\mathbf{c}}_i(\x_i,\!\mathbf{v}_i) 
		\!
		&=
		\!
		\mathbf{G}_i(\x_i)^T
		\!\!
		\left(\!
			\mathbf{c}_i(\x_i, \! \mathbf{G}_i(\x_i)\mathbf{v}_i) 
			\!
			+
			\! \!
			\mathbf{B}_i(\x_i)\dot{\mathbf{G}}_i(\x_i)\mathbf{v}_i
		\!\right)
		\\
	\label{eq:B_tilde}
		\widetilde{\mathbf{B}}_i(\x_i) 
		&= 
		\mathbf{G}_i(\x_i)^T\mathbf{B}_i(\x_i)\mathbf{G}_i(\x_i)
	\end{align}
\end{proposition}

Therefore, if the vector $\mathbf{u}_i$ is obtained to control the system with dynamics~\eqref{eq:agent_dynamics}, relation~\eqref{eq:linearizing_feedback} allows for the computation of the input $\boldsymbol{\tau}_i$ for the system~\eqref{eq:general_dynamics}. 

%TODO: mettere riferimento a Siciliano per modelli agenti? (Giovanni)

Note that the dynamics~\eqref{eq:agent_dynamics} along with~\eqref{eq:agent_pose} allow for the description of a sufficiently wide variety of vehicle dynamic models, thus enhancing team heterogeneity.

The following definitions will be used in the remainder of the paper: $\x = \ver{\x_i}$, $\mathbf{v} = \ver{\mathbf{v}_i}$, $\mathbf{z} = [\x^T, \mathbf{v}^T]^T$.

\section{Persistent Coverage Control}

\subsection{Problem Statement}

The \emph{persistent coverage} task deals with the recursive exploration of the operational area: the information gathered by the agents through their sensors becomes obsolete as time passes, thus requiring the team to visit regions of $Q$ where the amount of actual information has faded. 

To model this phenomenon, first we define the function $I \colon \RealSet_0^+\times Q \to \RealSet_0^+$, that quantifies the amount of useful information available at time $t$ in $\q \in Q$. The following equation models the \emph{information decay process}:
\begin{equation}
\label{eq:information_dynamics}
	\dot{I}(t,\q) 
	= 
	\delta I(t,\q) 
	+ 
	d(\p,\q)
	, \quad 
	\delta < 0
\end{equation}
i.e. the agents through the \ac{CTDF} give a positive contribution to information gathering, while the information decay rate $\delta$ yields its degradation. Denoting with $C^* > 0$ the desired level of information that should be maintained over $Q$, the \ac{TDF} describing the persistent coverage task is:
\begin{equation}
\label{eq:persistent_coverage_TDF}
	d_*(t,\q) = \max\left\{0, C^* - I(t,\q) \right\}
\end{equation}
while the \ac{TEF} models the difference between regions insufficiently covered and the current coverage provided by the team. Agents are expected to move towards its minimization.

To quantify how effectively the task is being fulfilled, the error index $\xi \colon \RealSet_0^+\times \mathcal{C}(Q)^N \to \RealSet_0^+$ is introduced:
\begin{equation}
	\xi(t,\p) = \int_Q f(e(t,\p,\q))\sigma(\q)\de \q
\end{equation}
where $f \colon \RealSet \to \RealSet_0^+$ is a \emph{penalty function} defined as $f(e) = \max\{0, e \}^p$ with $p = \{2,3,...\}$, and $\sigma \colon Q \to \RealSet_0^+$ is a weight that specifies the point importance in the environment.

The task is properly accomplished if the error function is kept as low as possible.

\begin{remark}
	The penalty function $f(\cdot)$ is continuous in $\RealSet$ and strictly convex in $\RealSet^+$, along with $\partial^n f/\partial e^n$ for $n<p$. 
\end{remark}	
	
\begin{remark}
	\label{rem:bounded_xi}
	Note that due to the definition of the \ac{TDF} in~\eqref{eq:persistent_coverage_TDF}, the error index $\xi(\cdot)$ is always bounded: 
	\begin{equation}
		0 \leq \xi(t,\p) \leq \int_{Q}f(C^*)\sigma\de \q
	\end{equation}
	for $t \geq t_0$.
\end{remark}
	
\subsection{Obstacles and Collision Avoidance}

Besides covering, agents are expected to avoid collisions with others, as well as with obstacles in the environment. Hence, the control action must guarantee that the distance between the agent and either another agent or an obstacle does not drop below a \emph{safety threshold} $r > 0$. We make the following assumption on the agents detection capabilities:

\begin{assumption}
	Each agent is able to detect other agents or obstacles when the relative distance is under a \emph{detection range} $R > r$. 
\end{assumption}

A team deployment is \emph{safe} if $\p \in \mathcal{P}(r)$, where:
\begin{multline}
	\mathcal{P}(x) 
	= 
	\left\{ 
		\p \in \mathcal{C}(\RealSet^n)^N 
		\colon
		\rho_{i,j} > x 
		\; \land \; 
		\rho_{i,k}^o > x \,
		\right. 
		\\ 
		\left.
		\forall i,j \in \mathcal{T}
		,\, 
		i\ne j
		,\, 
		\forall k \in \mathcal{O}
	\right\}
\end{multline}
with $\rho_{i,j}$ and $\rho_{i,k}^o$ denoting the distances of agent $i$ from agent $j$ and from the $k$-th obstacle point closest to it, the latter given by $\mathbf{c}_{k} \colon \mathcal{C}(Q) \to Q_{obs}$:
\begin{gather}
	\rho_{i,j}
	=
	\left\| \mathbf{S}(\p_i - \p_j) \right\|
	,\quad
	\rho_{i,k}^o 
	= 
	\left\| \mathbf{S}\p_i - \mathbf{c}_{k}(\p_i) \right\|
	\\
	\mathbf{c}_{k}(\p_i)
	= 
	\argmin_{\q \in Q_\text{obs}}\| \mathbf{S}\p_i - \q  \|
\end{gather}
The matrix $\mathbf{S}\in \RealSet^{n\times \dim \{\mathcal{C}(Q)\}}$ extracts the position components from the agent's pose vector.

To guarantee the team safety, the following  \emph{collision/obstacle avoidance function} is defined:
\begin{align}
	v(\p) 
	&=
	v^c(\p) + v^o(\p)
	\\
	&=
	\sum_{i \in \mathcal{T}} 
	v_i^c(\p) 
	+ 
	\sum_{i \in \mathcal{T}} 
	v_i^o(\p_i)
	\\
	&=
\label{eq:avoidance_function}
	\sum_{i \in \mathcal{T}} 
		\sum_{ j\in\mathcal{T}\backslash\{i\} }
		\!\! 
		l( \rho_{i,j} ) 
	+
	\sum_{i \in \mathcal{T}} 
		\sum_{k \in \mathcal{O}}
		l( \rho_{i,k}^o )
\end{align} 
where $l \colon \RealSet\rightarrow\RealSet_0^+$ is adapted from~\cite{Stipanovic2007}:
\begin{equation}
\label{eq:ell_function}
	l(x) 
	= 
	\left(
		\min\left\{ 0, \frac{x^2-R^2}{x^2-r^2} \right \}
	\right)^2
\end{equation}

\begin{remark}
	The function $l(\cdot)$ is such that if $x \ge R$, then $l(x) = 0$, while for $x \in \left(r, R\right)$ the function is strictly decreasing, and $\lim_{x\rightarrow r^+}l(x) = +\infty$. Given the avoidance function definition in~\eqref{eq:avoidance_function}, collisions do not occur and $\p(t) \in \mathcal{P}(r)$ is always verified, as long as $v(\cdot)$ attains finite values.
\end{remark}

\subsection{Control Law Definition}

The following control law is chosen for the agents:
\begin{equation}
\label{eq:control_law}
	\mathbf{u}_i(t,\x_i,\p) 
	= 
	\mathbf{u}_{\xi,i}(t,\x_i,\p) 
	+ 
	\mathbf{u}_{v,i}(\x_i,\p) 
	- 
	\mu \mathbf{v}_i
	,\,\,
	\mu > 0
\end{equation}
where:
\begin{align}
	\mathbf{u}_{\xi,i}(t,\x_i,\p) 
	&= 
	- 
	\beta 
	\mathbf{G}_i(\x_i)^T
	\frac{\partial \h_i(\x_i)}{\partial \x_i}^T
	\!\!
	\frac{\partial \xi(t,\p)}{\partial \p_i}^T
	\!\!\!,\,
	&\beta > 0
	\\
	\mathbf{u}_{v,i}(\x_i, \p) 
	&= 
	-
	\gamma 
	\mathbf{G}_i(\x_i)^T
	\frac{\partial \h_i(\x_i)}{\partial \x_i}^T
	\!\!
	\frac{\partial v(\p)}{\partial \p_i}^T
	\!\!\!,\, 
	&\gamma > 0
\end{align}
While $\mathbf{u}_{\xi,i}(\cdot)$ accelerates agent $i$ towards regions where $I(\cdot)$ is lower, thus inducing a minimization of the error function, $\mathbf{u}_{v,i}(\cdot)$ guarantees that it does not collide neither with other agents nor with obstacles. Finally, the term $-\mu \mathbf{v}_i$ acts as a damping term for the control law. 
\begin{theorem}
\label{th:Persistent_Coverage}
	If the initial team deployment is safe, i.e. $\p(t_0) \in \mathcal{P}(r)$, then, under the control law~\eqref{eq:control_law}, the team performs the persistent coverage task and no collision occurs, i.e. $\p(t) \in \mathcal{P}(r)$ for all $t \geq t_0$.
\end{theorem}
\begin{proof}
	Consider the following function:
	\begin{equation}
	\label{eq:V}
		V(t,\mathbf{z}) 
		= 
		\xi(t,\p(\x)) 
		+ 
		\frac{\gamma}{\beta}v(\p(\x)) 
		+ 
		\frac{1}{2\beta}\sum_{i \in \mathcal{T}}\left\|\mathbf{v}_i\right\|^2
	\end{equation}
	If $V(\cdot)$ attains finite values, and the agents act in order to decrease it, then the persistent coverage task is safely executed. Note that if the agents starts from a safe deployment, i.e. $\p(t_0) \in \mathcal{P}(r)$, then $V(\cdot)$ is finite at $t_0$. 
	
	Differentiation of~\eqref{eq:V} with respect to time yields:
	\begin{equation}
	\label{eq:V_dot_1}
		\dot{V} 
		= 	
		\frac{\partial \xi}{\partial t} 
		+
		\sum_{i \in \mathcal{T}} 
		\left[ 
			\left(  
				\frac{\partial\xi}{\partial \p_i} 
				+ 
				\frac{\gamma}{\beta}\frac{\partial v}{\partial \p_i}
			\right)\frac{\partial \h_i}{\partial \x_i}\mathbf{G}_i 
			+ 
			\frac{1}{\beta}\dot{\mathbf{v}}_i^T 
		\right] 
		\mathbf{v}_i
	\end{equation}
	Since $\dot{\mathbf{v}}_i = \mathbf{u}_i$ (see~\eqref{eq:agent_dynamics}), substitution of~\eqref{eq:control_law} in~\eqref{eq:V_dot_1} gives:
	\begin{equation}
		\dot{V} 
		= 
		\frac{\partial \xi}{\partial t} 
		- 
		\frac{\mu}{\beta}\sum_{i \in \mathcal{T}}\left\|\mathbf{v}_i\right\|^2
	\end{equation}
	The term $\partial \xi/\partial t$ can be further expanded obtaining:
	\begin{align}			
		\dot{V} 
		&= 
		- 
		\int_Q \frac{\partial f}{\partial e}\, \dot{I}\, \sigma\,\de \q 
		- 
		\frac{\mu}{\beta}\sum_{i \in \mathcal{T}} \left\| \mathbf{v}_i\right\|^2 
		\\
		&= 
		-
		\int_Q \frac{\partial f}{\partial e}\, \delta I\,\sigma\, \de \q 
		- 
		\int_Q \frac{\partial f}{\partial e}\,d\,\sigma\, \de \q 
		- 
		\frac{\mu}{\beta} \sum_{i \in \mathcal{T}} \left\| \mathbf{v}_i\right\|^2
		\\
	\label{eq:dot_V_final}
		&= 
		E_0 + E_1 + E_2
	\end{align}
	Note that $E_0 \geq 0$, since $\delta < 0$, whereas $E_1 \leq 0$ and $E_2 \leq 0$. The sign of $E_0$ is the consequence of the information degradation that characterizes the persistent coverage task. This produces an increment of $V(\cdot)$ with time. However, note that the only term in~\eqref{eq:V} which is explicitly time-dependent is the error index $\xi(\cdot)$. Since $\xi(\cdot)$ is bounded (see Remark~\ref{rem:bounded_xi}), then the sign of $E_0$ cannot produce an unlimited growth in $V(\cdot)$. The agents contribution is expressed by the terms $E_1$ and $E_2$, which always give a negative contribution. Hence, the persistent coverage task is properly executed by the agents, and $V(\cdot)$ attains finite values, proving the safe execution of the task. 
	
\end{proof}

\section{Control Law Decentralization}

Apart from the terms $\mathbf{G}_i(\x_i)^T (\partial\h_i(\x_i)/\partial \x_i)^T$ and $-\mu\mathbf{v}_i$, which are fully known as long as the state $\mathbf{z}_i$ is observable and the vehicle model is known, agent $i$ needs to compute the gradients of the error $\xi(\cdot)$ and the avoidance functions terms $v^c(\cdot)$ and $v^o(\cdot)$ in order to obtain $\mathbf{u}_i(\cdot)$. It will be shown that such quantities can be obtained in a distributed way without compromising the task fulfillment. 

To this end, we introduce the proximity graph $\mathcal{G}(t)$, with the agents as nodes. An edge between two agents exists if they are \emph{neighbors}, that is $\rho_{i,j} \le R_\text{com}$, with $R_\text{com}$ denoting agents \emph{communication range}. The set of neighbors of agent $i$ at time $t$ will be denoted with $\mathcal{N}_i(t)\subseteq\mathcal{T}$. In the following, we will consider the next two assumptions holding:

\begin{assumption}
	\label{assumption:graph_connected}
	The graph $\mathcal{G}(t)$ is always connected.
\end{assumption}

\begin{assumption}
	\label{assumption:com_range} 
	Let $r^\text{cov}_i = \max\{\| \q_i - \q\|  \colon  d_i(\p_i,\q) > 0\}$ and $\overline{r}^\text{cov} = \max_{i\in \mathcal{T}} \{r^\text{cov}_i\}$. We assume that $R_\text{com} > 2\overline{r}^{cov}$. As a result, agents with overlapping \ac{ADF}s are neighbors.
\end{assumption}

\subsection{Collision and Obstacle Avoidance} 

To enforce the collision and obstacle avoidance policy agent $i$ must compute the gradient of $v(\cdot)$, that is:
%
%\begin{align}
%	\frac{\partial v(\p)}{\partial \p_i} 
%	&= 
%	2\frac{\partial v_i^c(\p)}{\partial \p_i} 
%	+ 
%	\frac{\partial v_i^o(\p_i)}{\partial \p_i}
%	\\
%	&=
%	2 \sum_{j \in \mathcal{T}/\{i\}}
%	\frac{\partial l(\rho_{i,j})}{\partial \rho_{i,j}}
%	\frac{\partial \rho_{i,j}}{\partial \p_i}
%	+
%	\sum_{k \in \mathcal{O}}
%	\frac{\partial l(\rho_{i,k}^o)}{\partial \rho_{i,k}^o}
%	\frac{\partial \rho_{i,k}^o}{\partial \p_i}
%\end{align}
\begin{equation}
	\frac{\partial v(\p)}{\partial \p_i} 
	=
	2 \!\!\!\sum_{j \in \mathcal{T}/\{i\}}
	\frac{\partial l(\rho_{i,j})}{\partial \rho_{i,j}}
	\frac{\partial \rho_{i,j}}{\partial \p_i}
	+
	\sum_{k \in \mathcal{O}}
	\frac{\partial l(\rho_{i,k}^o)}{\partial \rho_{i,k}^o}
	\frac{\partial \rho_{i,k}^o}{\partial \p_i}
\end{equation}
where:
\begin{equation}
	\frac{\partial l(x)}{\partial x}
	=
	\left\{
		\begin{array}{ll}
			0, & \,\, x > R  \,\, | \,\, x < r
			\\
			4\frac{\left(R^2 - r^2\right)\left(x^2 - R^2\right)}{\left(x^2 - r^2\right)^3} x, & \,\, R \geq x > r
			\\
			\text{undefined}, & \,\, x = r
		\end{array}
	\right.
\end{equation} 
\begin{equation}
	\frac{\partial\rho_{i,j}}{\partial \p_i} 
	= 
	\frac{\left(\p_i - \p_j\right)^T\mathbf{S}^T\mathbf{S}}{\rho_{i,j}}
	,\quad
	\frac{\partial\rho_{i,k}^o}{\partial \p_i} 
	= 
	\frac{\left(\mathbf{S}\p_i - \mathbf{c}_k(\p_i)\right)^T\mathbf{S}}{\rho_{i,k}^o}
\end{equation}

Therefore, the collision and obstacle avoidance function gradients must be computed only when another agent $j$ or an obstacle $k$ are within the agent $i$ detection radius, that is $\rho_{i,j} \leq R$ and $\rho_{i,k}^o \leq R$, respectively. Thus, the collision/obstacle avoidance is implicitly distributed.

\subsection{Coverage Control and Information Level Estimation}
\label{sec:CoverageControlDecentralization}

To accomplish the coverage task, agents need to compute the quantity $\partial \xi /\partial \p_i$. Using the chain rule and observing that the \ac{ADF}s are dependent on the respective agent's pose only, the following holds:
\begin{equation}
	\frac{\partial \xi}{\partial \p_i} 
	= 
	-
	\int_{Q} 
	\frac{\partial f}{\partial e}
	\frac{\partial d}{\partial \p_i}
	\sigma \de \q
	= 
	-
	\int_{Q_i} 
	\frac{\partial f}{\partial e}\frac{\partial d_i}{\partial \p_i}
	\sigma \de \q
\end{equation}
where $Q_i(\p_i) = \supp{\partial d_i(\p_i,\cdot)/\partial \p_i}$. The weight $\sigma$ is assumed known to each agent prior to the deployment. The gradient $\partial d_i/\partial \p_i$ can be computed autonomously by each agent. To obtain the term $\partial f/\partial e$ the knowledge of the \ac{TEF} $e(t,\p,\q)$ on $\q \in Q_i(\p_i)$ is needed. We recall that $e(t,\p,\q) = d_*(t,\q) - d(\p,\q)$, see~\eqref{eq:TEF}. 

Because of Remark~\ref{remark:support_inclusion} and Assumption~\ref{assumption:com_range}, if all agent $i$'s neighbors share with it their poses and their \ac{ADF} parameters, agent $i$ is able to compute the \ac{CTDF} $d(\p,\q)$ on $\q \in Q(\p_i)$.

The computation of $d_*(t,\q)$ requires the knowledge of the attained level of information $I(t,\q)$. We now introduce an algorithm for the decentralized estimation of $I(t,\q)$, adapted from~\cite{PalaciosGasos2016}. Whereas it was originally formulated in a discrete time context, the algorithm is here rearranged to fit in a continuous time scenario. The conventions of the descriptor function framework will be employed. For the sake of compactness, the quantities' dependence on $\p$ and $\q$ will be dropped when this will not compromise clarity. 

Each agent autonomously computes a continuous time estimation $\hat{I}_i(t)$, which is shared with the neighbors periodically  during the \emph{update instants} $t_k = kT$, with $k = \{1,2,\dots\}$. Given~\eqref{eq:information_dynamics}, it follows that for $t\in \left[t_{k-1}, t_k\right)$:
\begin{equation}
	I(t) 
	= 
	e^{\delta (t-t_{k-1})}I(t_{k-1}) 
	+ 
	\sum_{i\in\mathcal{T}}\int_{t_{k-1}}^t e^{\delta (t-\tau)} d_i(\tau)\de \tau
\end{equation}
Since no communication occurs between update instants, agent $i$ estimates $I(t)$ considering only its contribution to the task:
\begin{equation}
\label{eq:information_time_history}
	\hat{I}_i(t) 
	= 
	e^{\delta (t-t_{k-1})}\hat{I}(t_{k-1}) 
	+ 
	\int_{t_{k-1}}^t \!\! e^{\delta (t-\tau)}d_i(\tau)\de \tau,
\end{equation}
Then, at each update instant $t_k$ the agent communicates to the neighbors its present estimation $\hat{I}_i(t_k)$:
\begin{equation}
	\hat{I}_i(t_k) = e^T\hat{I}_i(t_{k-1}) + \tilde{d}_i(t_k)
	,\quad
	\tilde{d}_i(t_k) = \int_{t_{k-1}}^{t_k} \!\! e^{\delta(t-\tau)}d_i(\tau)\de \tau
\end{equation}
After the agents exchange their estimates, the first correction is performed: 
\begin{itemize}
	\item for all $\q \in \supp{\tilde{d}_i(t_k)}$
	\begin{equation}
		\hat{I}_i^-(t_k) 
		= 
		\hat{I}_i(t_k) 
		+ 
		\!\!
		\sum_{j\in\mathcal{N}_i(t_k)}
		\!\!
		\max\left\{0, \,\, \hat{I}_j(t_k) - e^T\hat{I}_i(t_{k-1}) \right\}
	\end{equation}	
	\item whereas for all $\q \notin \supp{\tilde{d}_i(t_k)}$
	\begin{equation}
		\hat{I}_i^-(t_k) 
		= 
		\max_{j\in\mathcal{N}_i(t_k)}\left\{\hat{I}_j(t_k), \,\, e^T \hat{I}_i(t_{k-1})\right\}
	\end{equation}
\end{itemize}

At this point, since for all $\q \notin \supp{\tilde{d}_i(t_k)}$ only the highest contribution between agent $i$'s coverage (which, by definition, results from the time degradation of the one attained at time $t_{k-1}$) and those of neighbors is considered, agent $i$ obtains a raw coverage underestimation. In fact, potential overlappings between $\tilde{d}_j(t_k)$ and $\tilde{d}_l(t_k)$, with $j,l \in \mathcal{N}_i(t_k)$, which yield a higher coverage level than the one currently estimated for all $\q \notin \supp{\tilde{d}_i(t_k)}$, would be ignored. However, such information can be easily retrieved through a second correction, which improves each agent's estimation. In fact, having received $\hat{I}_j(t_k)$ from neighbors, agent $i$ is able to compute the portion of $\tilde{d}_i(t_k)$ overlapped with those of neighbors. Let 
\begin{equation}
	O_i(t_k) 
	= 
	\bigcup_{j\in \mathcal{N}_i(t_k)} 
	\left\{
		\q \in Q 
		\colon 
		\hat{I}_j(t_k) - e^T \hat{I}_i(t_{k-1}) > 0 
	\right\}
\end{equation}
be such overlapped area and define 
\begin{equation}
	\tilde{d}_i^o(t_k) 
	= 
	\left\{
	\begin{array}{ll}
		\tilde{d}_i(t_k), & \q \in O_i(t_k)\\
		0, & \q \notin O_i(t_k)
	\end{array} 
	\right.
\end{equation}
Neighbors then exchange their $\tilde{d}_i^o(t_k)$ and perform the last correction:
\begin{itemize}
	\item for all $\q\in \supp{\tilde{d}_i(t_k)}$
	\begin{equation}
		\hat{I}_i(t_k) = \hat{I}_i^-(t_k), 
	\end{equation}
	
	\item whereas for all $\q\notin \supp{\tilde{d}_i(t_k)}$
	\begin{equation}
		\hat{I}_i(t_k) 
		= 
		\hat{I}_i^-(t_k) 
		+ 
		\!\!
		\sum_{j\in\mathcal{N}_i(t_k)}
		\!\!\!
		\tilde{d}_j^o(t_k) 
		-
		\max_{j\in\mathcal{N}_i(t_k)}
		\!\!
		\tilde{d}_j^o(t_k) 
	\end{equation}
\end{itemize}

% Observe that, agent $i$ does not need to compute the convolution integral appearing in \eqref{eq:information_time_history} in order to obtain $\tilde{d}_i(t_k)$: since $\tilde{d}_i(t_k) = \hat{I}_{com,i}(t_{k}) - e^{\delta T_{com}}\hat{I}_i(t_{k-1})$, only the complete integral $\hat{I}_{com,i}(t_k)$ and the last update $\hat{I}_i(t_{k-1})$ are needed.
 
The following result ensures the existence of regions where the estimation error is zero:
\begin{theorem}
\label{th:Estimation}
	Assume that $\hat{I}_i(0) = I(0)$ for all $\q \in Q$, and for all $i \in \mathcal{T}$. In addition, assume that each agent can travel for a maximum distance $\ell^\text{max}_T$ during a period $T$. Then, there exists a region centered in the agent position at time $t_k$, $\mathbf{S}\p_i(t_k)$, and with radius $r^* = R_\text{com} - \overline{r}^\text{cov} - (N-1)\ell^\text{max}_T$ denoted with:
	\begin{equation}
		Z_i(t_k) = \left\{\q \in Q \colon \| \mathbf{S}\p_i(t_k) - \q\| \le r^* \right\}
	\end{equation}
	such that $\hat{I}_i(t_k) = I(t_k)$ for all $\q \in Z_i(t_k)$, and $i \in \mathcal{T}$. 
\end{theorem}

\begin{proof}
	The proof consist in a transposition of Theorem~IV.6 proof in~\cite{PalaciosGasos2016}, which is proved in a discrete time scenario, to the continuous time case considered in this paper. Theorem~IV.6 in~\cite{PalaciosGasos2016} states that each agent's estimation is correct within a distance $R_\text{com} - \overline{r}^\text{cov} - (N-1)u^\text{max}$ at each discrete time instant, with $u^\text{max}$ being the maximum allowable control input norm for agents characterized by single-integrator kinematics ($\x(t+1) = \x(t) + \mathbf{u}(t)$), assuming that Assumptions~\ref{assumption:graph_connected} and~\ref{assumption:com_range} hold and that $\hat{I}_i(0) = I(0)$ for all $i \in \mathcal{T}$. In our case, the maximum allowed travel distance $\ell^\text{max}_T$ during a period $T$ is the continuous time analogous of the discrete time maximum control input norm $u^\text{max}$ in \cite{PalaciosGasos2016} for single integrators.
\end{proof}

\subsection{Decentralized Control Law}

The decentralized control law is then defined as follows:
\begin{equation}
\label{eq:control_law_decentralized}
	\hat{\mathbf{u}}_i(t,\x_i,\p) 
	= 
	\hat{\mathbf{u}}_{\xi,i}(t,\x_i,\p_i) 
	+ 
	\mathbf{u}_{v,i}(\x_i,\p) 
	- 
	\mu \mathbf{v}_i, \,\, \mu > 0
\end{equation}
where the coverage term is computed as follows:
\begin{equation}\label{eq:control_newestimation}
	\hat{\mathbf{u}}_{\xi,i}(t, \x_i, \p_i) 
	= 
	-\beta \mathbf{G}_i(\x_i)^T
	\frac{\partial \h_i(\x_i)}{\partial \x_i}^T
	\!\!
	\frac{\partial \hat{\xi}_i(t,\p_i)}{\partial \p_i}^T
	\!\!, \,\,
	\beta > 0
\end{equation}
The task error index is estimated by each agents using the distributed estimation discussed in Section~\ref{sec:CoverageControlDecentralization}:
\begin{equation}
	\frac{\partial \hat{\xi}_i(t,\p_i)}{\partial \p_i} 
	= 
	-
	\int_{Q_i(\p_i)} 
	\!\!
	\frac{\partial f(\hat{e}(t,\q))}{\partial \hat{e}(t,\q)}
	\frac{\partial d_i(\p_i,\q)}{\partial \p_i} 
	\sigma\de\q
\end{equation}
where the \ac{TEF} estimation is given by:
\begin{align}
	\hat{e}_i(t,\q) 
	&= 
	\hat{d}_{*,i}(t,\q) 
	- 
	\!\!
	\sum_{j\in\mathcal{N}_i(t)}
	\!\!
	d_j(\p_j(t),\q) 
	-
	d_i(\p_i(t),\q)
	\\
	\hat{d}_{*,i}(t,\q) 
	&= 
	\max \left\{0,C^* - \hat{I}_i(t,\q)\right\}
\end{align}

Observe that in the definition of $\hat{e}_i(\cdot)$ we assumed that agent $i$ has the instantaneous knowledge of the \ac{ADF}s of its neighbors. This is possible as long as agents are able to exchange their poses and \ac{ADF}s parameters at a sufficiently high rate. Therefore, it should be pointed out that two types of communication protocols need to be enforced. The first one, allowing estimation updates as described in Section \ref{sec:CoverageControlDecentralization}, must be guaranteed with period $T$ and requires the exchange of high amounts of data, since each agent should, at least,  send its own estimated coverage map and receive those of neighbors. The second one requires communication of quantities, such as positions, orientations, and \ac{ADF}-related parameters that are fast and easy to handle and to exchange.

\subsection{Comments on the Information Level Estimation}
\label{sec:comments}

With reference to Theorem~\ref{th:Estimation}, the expression of $r^*$ implies that, if a guarantee on the quality of the coverage estimation is to be sought, communication updates should happen more frequently if vehicles are allowed to move faster (and thus have larger $\ell^\text{max}_T$). Agents speed can be reduced by choosing higher values of the damping $\mu$ in~\eqref{eq:control_law}. 

Moreover, in a team with a high number of agents, the region around each agent where, at time $t_k$, its estimation is correct is sensibly reduced. The reason is twofold: first, the cardinality $N$ of the team directly affects the radius $r^*$; furthermore, the estimator definition in~\eqref{eq:information_time_history}, which considers only the agent contribution between the update instants,  does not take into account the contributions of the other agents, thus having higher chances of errors along with a higher number of team members.

In fact, the level of information provided by~\eqref{eq:information_time_history} is an underestimation of the real attained level:
\begin{equation}
	\label{eq:underestimation}
	\hat{I}_i(t,\q) \le I(t,\q),
	\quad
	t \geq t_0,
	\quad
	\forall i\in \mathcal{T}
	,\,\,
	\forall \q \in Q
\end{equation}
for the aforementioned reasons.

Having shown that the centralized control law guarantees the task fulfillment, it should be pointed out that the decentralized scheme only induces a slight performance degradation, without compromising the mission. Indeed, since \eqref{eq:underestimation} holds, agent $i$ could be attracted by regions with an already high level of information $I$, with a remarkable disparity $\hat{I}_i \le I$. This would be a reasonable trade-off due to the employment of a distributed estimation. Conversely, agent $i$ would never be repulsed from regions with a very low level of true attained coverage, since $\hat{I}_i$ would be low in those regions as well. In addition, because of its definition, $\hat{\xi}_i(t,\p) \ge \xi(t,\p)$ for all $i \in \mathcal{T}$, and for $t \ge t_0$, thus remarking that with only locally achievable information agent $i$ overestimates the error function, assuming that the task is farther from being fulfilled than it truly is.

\section{Simulation Results}
\label{sec:SimulationResults}

%\begin{figure}
%	\centering
%	\subfloat[Gaussian \ac{ADF}.]{
%		\label{fig:GaussianADF}
%		\includegraphics[width=.45\columnwidth]{images/ADF_Gaussian.pdf}
%	}
%	\hfil
%	\subfloat[Gaussian \ac{ADF} with limited field of view.]{
%		\label{fig:GaussianADF_FOV}
%		\includegraphics[width=.45\columnwidth]{images/ADF_Gaussian_FoV.pdf}}
%	\caption{Agents \ac{ADF}s.}	
%\end{figure}

%TODO: riferimento a Siciliano per espressioni dinamica agenti corretto? Mettere anche numerodi equazione all'interno del libro? (Giovanni)

To validate the proposed distributed control law a persistent coverage scenario in a $100 \times 100$ square environment involving $6$ agents was simulated. The information decay rate was set to $\delta = -0.1$ and $C^*= \sigma(\q) = 1$ for all $\q \in Q$. Seven obstacles were placed in the environment: four of them are virtual walls surrounding the operational area, preventing agents from exiting it. The remaining ones are polygons placed within $Q$. Three of the agents have double integrator dynamics while the other three are dynamic unicycles with unitary mass and moment of inertia (for the equations see~\cite[Chap. 11.4]{Siciliano2009}). Having set $r = 0.5$ and $R = 3$, the initial poses were chosen such that $\p(t_0)\in\mathcal{P}(r)$. The double integrators carry isotropic \emph{Gaussian \ac{ADF}s}:
\begin{equation}
\label{eq:ADF_Gaussian}
	d_{G}(\p_i,\!\q) 
	\!
	=
	\!
	A \exp \left( 
		-\frac{1}{2}
		\left\|
			\!
			\begin{bmatrix}
				\cos\theta & \sin\theta
				\\
				-\sin\theta & \cos\theta
			\end{bmatrix}
			\!
			(\q - \mathbf{S}\p_i)
		\right\|^2_{\boldsymbol{\Sigma}^{-1}} \right)
\end{equation}
with $A = 3$ and $\Sigma = \diag{3, 3}$. The unicycles are characterized by a \emph{Gaussian \ac{ADF} with limited field of view}: 
\begin{equation}
	d_{G,\text{FOV}}(\p_i,\q) = d_{G}(\p_i,\q)f_\text{FOV}(\p_i,\q)
\end{equation}
with
\begin{align}
	&f_\text{FOV}(\p_i,\q) 
	= 
	f_{\text{FOV},r}(\p_i,\q) f_{\text{FOV},l}(\p_i,\q)
	\\
	&f_{\text{FOV},r}(\p_i,\q) 
	= 
	\left(1 + e^{-k\left(r_1\cos\left(\phi/2\right) + r_2\sin\left(\phi/2\right)\right)}\right)^{-1}
	\\
	&f_{\text{FOV},l}(\p_i,\q) 
	= 
	\left(1 + e^{-k\left(-r_1\cos\left(\phi/2\right) + r_2\sin\left(\phi/2\right)\right)}\right)^{-1}
	\\
	&r_1
	= 
	(q_x - x_i)\cos\left(\theta_i - \pi/2\right) + (q_y - y_i)\sin\left(\theta_i - \pi/2\right)
	\\
	&r_2
	= 
	-(q_x - x_i)\sin\left(\theta_i - \pi/2\right) + (q_y - y_i)\cos\left(\theta_i - \pi/2\right)\!\!
\end{align}
The function $f_\text{FOV}: \mathcal{C}(Q)\times \RealSet^2 \rightarrow [0,1]$ shapes the field of view and is equal to $1$ inside it, while smoothly decreasing to $0$ as its boundaries are approached. The parameter $k$ models the slope and was set to $2$, while $\phi$ is the field of view angle and was set to $\pi/2$. As regards the control gains, the following values were assigned: $\alpha = \beta = 30$, $\mu = 20$. The communication range is $R_\text{com} = 150$, thus ensuring the connectivity of the graph $\mathcal{G}(t)$. The estimate update period $T$ has been set to $1$.

\begin{figure}
	\centering
	\includegraphics[width=.9\columnwidth]{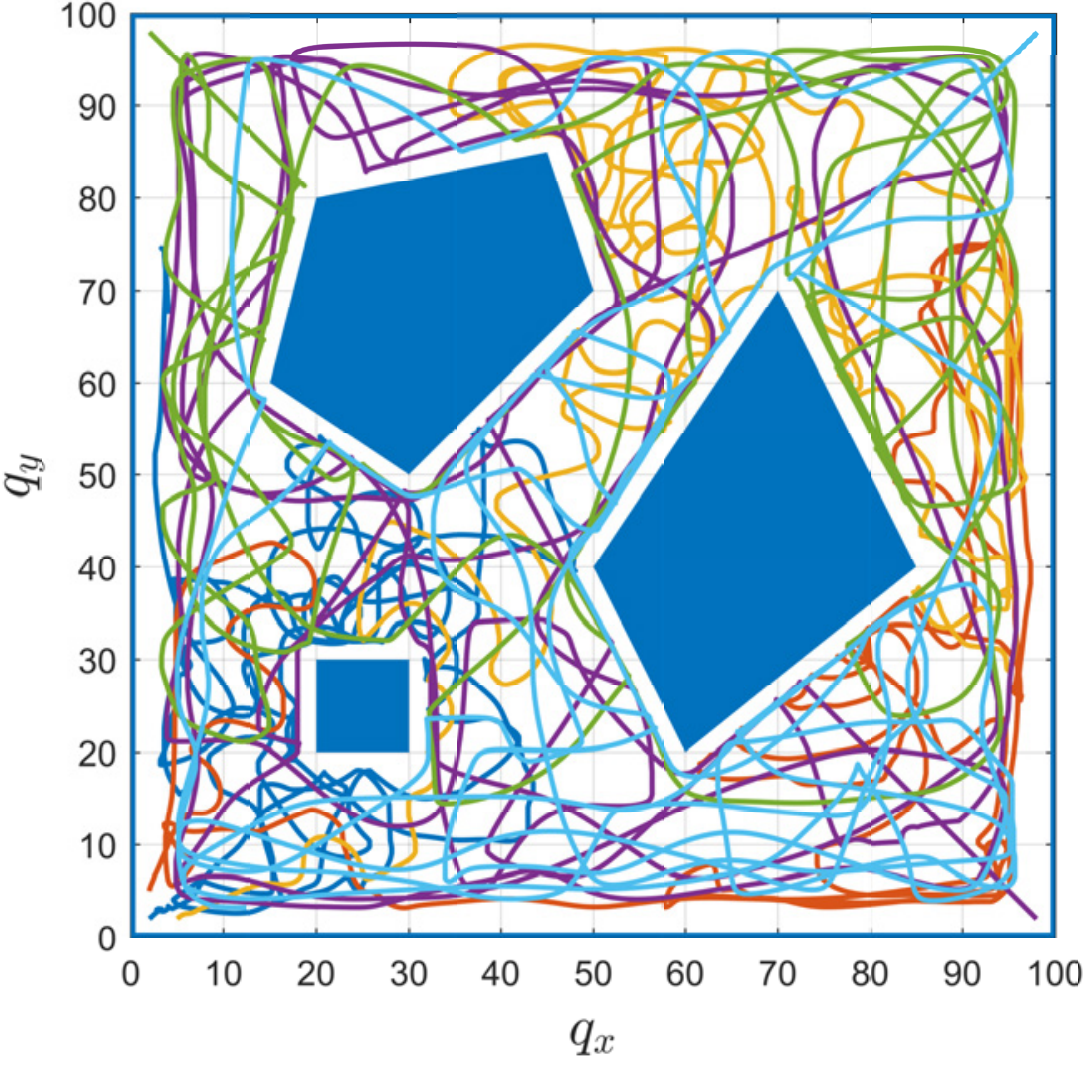}
	\caption{Agents trajectories.}
	\label{fig:AgentsTrajectories}
\end{figure}

\begin{figure}
	\centering
	\subfloat[Normalized true $\xi(t)$.]{
		\label{fig:XiEvolution}
		\includegraphics[width=\columnwidth]{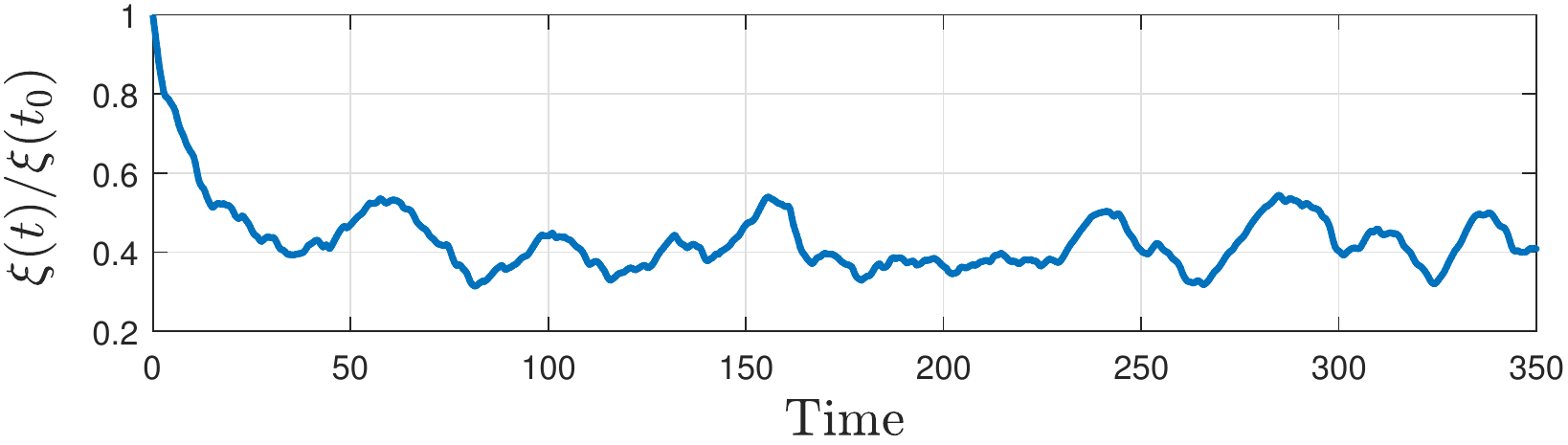}
	}
	\\
	\subfloat[{Normalized agents estimation error, for $t \in [100, 110]$.}]{
		\label{fig:XiEstimationError}
		\includegraphics[width=\columnwidth]{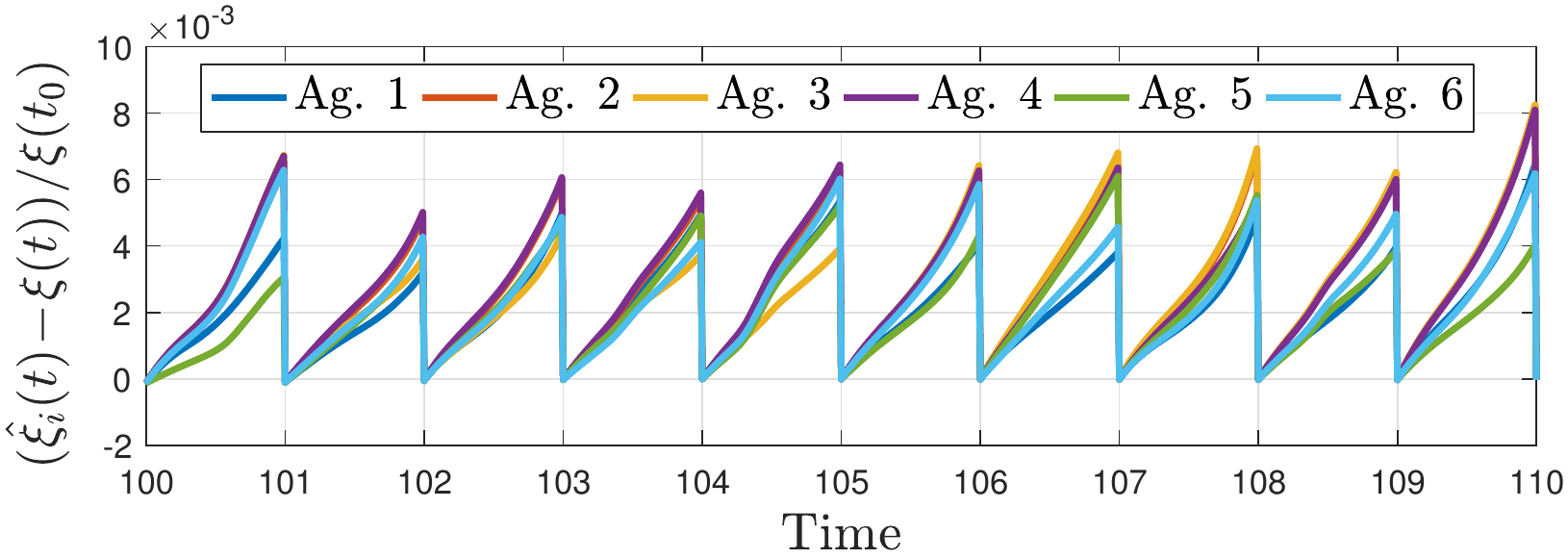}
	}
	\caption{Evolution of $\xi(t)$ and agents estimation error.}
\end{figure}

\begin{figure}
	\centering
	\subfloat[Inter-agent distances.]{
		\label{fig:InterAgentDistances}
		\includegraphics[width=\columnwidth]{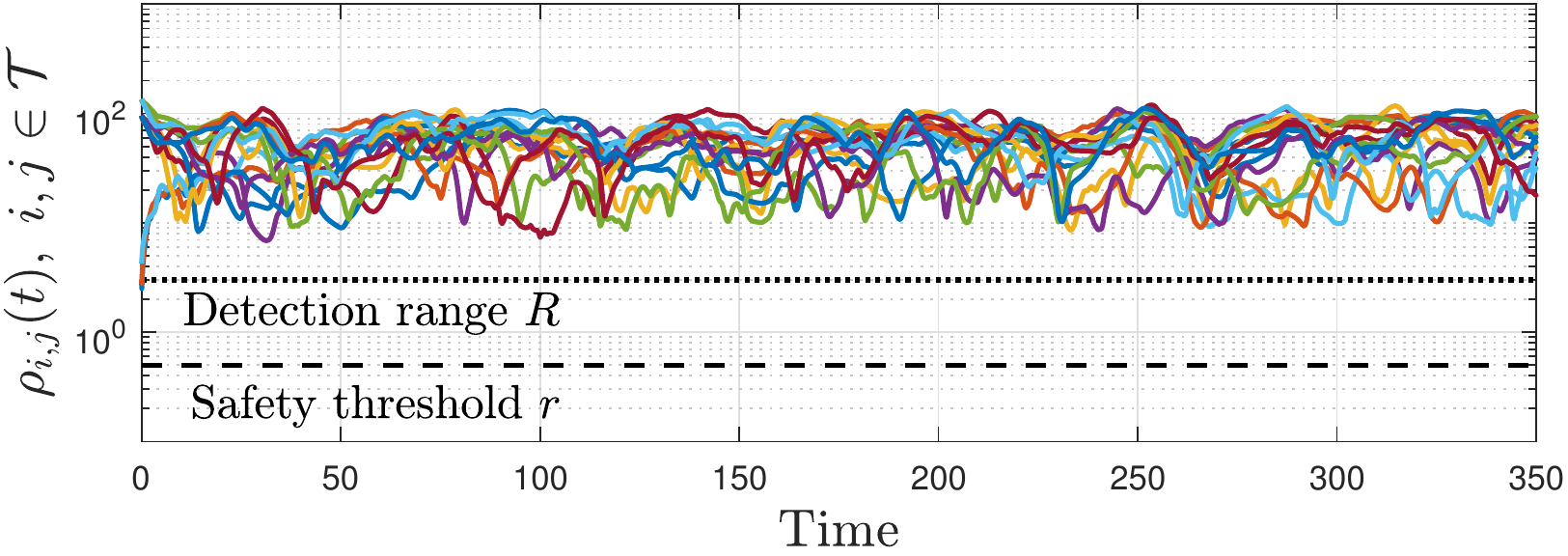}
	}
	\\
	\subfloat[Agent-obstacle distances.]{
		\label{fig:AgentObstacleDistances}
		\includegraphics[width=\columnwidth]{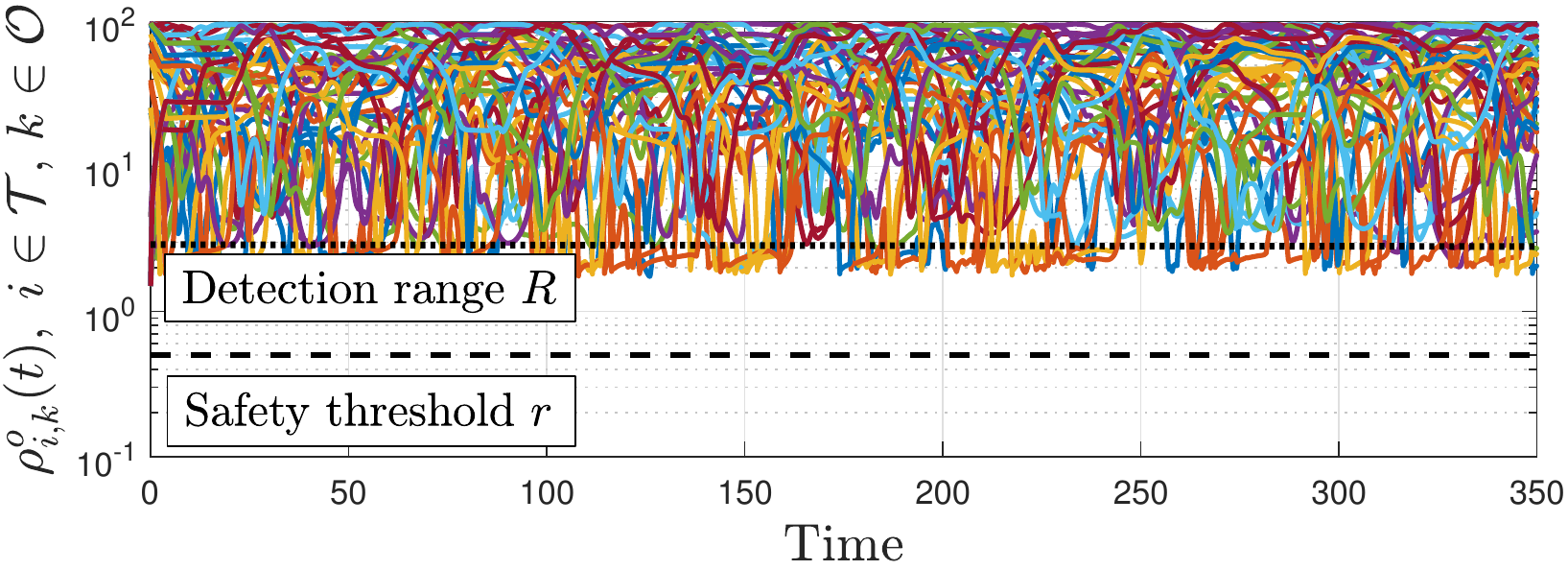}
	}
	\caption{Inter-agent and agent-obstacle distances evolution.}
\end{figure}

Agents trajectories are shown in Figure~\ref{fig:AgentsTrajectories}. It can be seen that the persistent task is efficiently performed and all the area is continuously covered by the agents. The efficiency is confirmed by the evolution of the error index $\xi(\cdot)$ shown in Figure~\ref{fig:XiEvolution}. The agents estimation error is shown in Figure~\ref{fig:XiEstimationError}. Note that $\hat{\xi}(t) \geq \xi(t)$ during the analyzed period, that is agents overestimate the error index confirming what pointed out in Section~\ref{sec:comments}. The error grows between the updated instants, and becomes approximately zero when updates are exchanged, proving the effectiveness of the distributed estimation algorithm. Figures~\ref{fig:InterAgentDistances} and~\ref{fig:AgentObstacleDistances} confirm that the task is safely executed, and no collisions occur.

\section{Conclusions}

We presented a distributed cooperative control strategy for persistent coverage tasks execution. The proposed control law is designed for coping with agents having different types of dynamics and heterogeneous sensing capabilities. The latter was modeled and handled by means of the tools provided by the descriptor function framework. Task execution safety was guaranteed by means of obstacles and inter-agent collision avoidance functions. An algorithm for the estimation of the attained level of coverage was derived in order to decentralize the control law. The effectiveness of the control law and of the estimation algorithm was formally proved and showed by means of a numerical simulation.

%\section*{Acknowledgments}

%TODO : accenti nella bibliografia

\bibliographystyle{IEEEtran}
\bibliography{./ieee_bibliography/IEEEabrv,./biblio}

\end{document}